\documentclass{article}
\usepackage[%
journal=JST,    
lang=american,   
]{ems-journal}

\usepackage{amsmath,amsfonts,
amsthm, color}
\usepackage{hyperref}
\usepackage{csquotes}

\usepackage[active]{srcltx}

\usepackage{mathrsfs}
\usepackage{multicol}

\newtheorem{theorem}{Theorem}[section]

\newtheorem{proposition}[theorem]{Proposition}
\newtheorem{corollary}[theorem]{Corollary}
\newtheorem{lemma}[theorem]{Lemma}

\theoremstyle{definition}
\newtheorem{remark}[theorem]{Remark}

\newtheorem{hypothesis}[theorem]{Hypothesis}

\def\R{\mathbb{R}}

 \def\Op{\mathfrak{Op}}

\def\x{\underline{x}}

\def\z{\mathfrak{z}}

\def\Z{\mathbb{Z}}

\def\bb1{{\textrm{1}\hspace{-3pt}\mathbf{l}}}
\def\Ie0{[-\epsilon_0,\epsilon_0]}

\def\BC2{\mathbb{B}\big(\mathbb{C}^2\big)}

\def\beq{\begin{equation}}
\def\eeq{\end{equation}}

\numberwithin{equation}{section}

\numberwithin{equation}{section}

\begin{document}

\title{Sharp spectral stability for a class of singularly perturbed pseudo-differential operators}
\titlemark{Sharp spectral stability}

\emsauthor{1}{Horia D. Cornean}{H.~Cornean}
\emsauthor{2}{Radu Purice}{R.~Purice}

\emsaffil{1}{Department of Mathematical Sciences, Aalborg University, Skjernvej 4A, DK-9220 Aalborg, Denmark \email{cornean@math.aau.dk}}

\emsaffil{2}{\enquote{Simion Stoilow} Institute of Mathematics of the Romanian Academy, P.O. Box 1-764, 014700 Bucharest, Romania \email{Radu.Purice@imar.ro}}

\classification[35S05]{81Q10, 81Q15}


\begin{abstract}
Let $a(x,\xi)$ be a  real H\"ormander symbol of the type $S_{0,0}^0(\R^{d}\times \R^d)$, let $F$ be a smooth function with all its derivatives globally bounded, and let $K_\delta$ be the self-adjoint Weyl quantization of the perturbed symbols $a(x+F(\delta\, x),\xi)$, where $|\delta|\leq 1$. First, we prove that the Hausdorff distance between the spectra of $K_\delta$ and $K_{0}$ is bounded by $\sqrt{|\delta|}$, and we give examples where spectral gaps of this magnitude can open when $\delta\neq 0$. Second, we show that the distance between the spectral edges of $K_\delta$ and $K_0$ (and also the edges of the inner spectral gaps, as long as they remain open at $\delta=0$) are of order $|\delta|$, and give a precise dependence on the width of the spectral gaps.
\end{abstract}

\maketitle





\begin{funding}
This work was partially supported by the Grant 8021-00084B of the Independent Research Fund Denmark $|$ Natural Sciences.
\end{funding}


\section{Introduction and main results}

Let $a(x,\xi)$ be a real H\"ormander symbol \cite{H-3} of class $S^0_{0,0}(\R^{d}\times\R^d)$, i.e. a smooth function on $\R^{2d}$ satisfying the estimate:
\begin{align}\label{dc1}
\sup_{x,\xi \in \R^d}|D_x^\alpha D_\xi^\beta a|<\infty,\quad \forall \alpha,\beta\in \mathbb{N}^d.
\end{align}
For $|\delta|\leq 1$ let $a_\delta(x,\xi)=a(\sqrt{1+\delta}\, x,\sqrt{1+\delta}\, \xi)\in \R$. It belongs to the same class. 
 We denote by $H_\delta=\Op^w (a_\delta)$ the self-adjoint operator generated by the Weyl quantization, which means that:
\[\langle \psi,H_\delta \phi\rangle :=(2\pi)^{-d}\int_{\R^d}d\xi\, \int_{\R^{2d}}\;dx'\; dx\,  e^{i\xi\cdot(x-x')}a_\delta((x+x')/2,\xi)\overline{\psi}(x)\phi(x'),\]
where $\psi,\phi\in \mathscr{S}(\R^d)$ and  $\langle\cdot\,,\,\cdot\rangle $ denotes the usual scalar product in $L^2(\R^d)$ (considered to be anti-linear in the first variable). The operator $H_\delta$ has a distribution kernel that can be written as the oscillatory integral: 
\begin{equation}\label{hc2}
    \mathfrak{H}_\delta(u,v):=(2\pi)^{-d} \int_{\R^d} a_\delta(u,\xi)e^{i \xi\cdot v}d\xi,\quad u=(x+x')/2,\; v=x-x'.
\end{equation}

In \cite{GRS}, Gr\"ochenig {\itshape et.al.} proved that the spectral edges of the spectrum $\sigma(H_\delta)$ are Lipschitz at  $\delta=0$.   The problem is non-trivial because the map $\delta\, \mapsto H_\delta$ is not necessarily differentiable in the operator norm topology, which can already be seen at the level of the symbol: $a$ should have some extra linear decay in both $x$ and $\xi$ in order to make sure that $H_\delta-H_0$ has a norm of order $\delta$, which would imply that the Hausdorff distance between the spectra of $H_\delta$ and $H_0$ is Lipschitz continuous at zero. 

Nevertheless, the authors of  \cite{GRS}  show that such a strong decay is far from necessary if one is only interested in the spectral edges. Actually they even consider more general operators corresponding to symbols of Sj\"{o}strand type, operators which belong to certain weighted modulation spaces; see \cite{G0} and references therein for an introduction to the subject. 

A similar  phenomenon  appears in the case of long range magnetic perturbations \cite{B2, Cornean10, CP1, CP2, C4}. In fact, the two problems are very much related, see Section \ref{sec3.2} of the current manuscript for more details.


\subsection{A more general perturbation} In this manuscript we are interested in a more general perturbation of the symbol, where the dilation treated in  \cite{GRS} becomes just a particular case. In order to achieve that, we have to "rotate" the operators $H_\delta$ in the following way:
\begin{lemma}\label{P-U-trsf}
Denote by $U_\delta$ the unitary transformation in $L^2(\R^d)$ given by 
\[(U_\delta f)(x)=(1+\delta)^{-d/4}f((1+\delta)^{-1/2}\; x),\quad\forall f\in L^2(\R^d).\]
Then $U_\delta^* H_\delta U_\delta$ equals the Weyl quantization  of the symbol $a(x+\delta\, x\,,\, \xi)$, and is isospectral with $H_\delta$.
\end{lemma}
The proof of this lemma, rather straightforward, will be given in the next section. 
 The advantage of working with symbols shifted only in $x$ is that we can identify a larger class of perturbations, where the same spectral results as proved in \cite{GRS} hold true. More precisely, instead of $x+\delta\, x$ we will consider $x+F(\delta\,x)$ where $F$ satisfies the following assumptions:

\begin{hypothesis}\label{hyp1.2}
Let $F\in [C^\infty(\R^d)]^d$ be a smooth real vector-valued function with all its derivatives of all order uniformly bounded (thus $F$  can grow linearly at infinity).  
Given any real symbol $a\in S^0_{0,0}(\R^d\times\R^d)$ as in \eqref{dc1} and $\delta\in \R$, let $K_\delta$ be the Weyl quantization of the symbol 
$a[F]_{\delta}(x,\xi):=a\big (x+F(\delta\, x),\xi\big )$, i.e. $K_\delta=\Op^w\big (a[F]_{\delta}\big)$. Also, the distribution kernel of $K_\delta$ (defined as in \eqref{hc2}) is denoted by  $\mathfrak{K}_\delta$. If $F(x)=x$, then $K_\delta$ and $H_\delta$ are unitarily equivalent and isospectral.
\end{hypothesis}

We will only work with symbols of class $S^0_{0,0}(\R^d\times\R^d)$ (included in the Sj\"{o}strand class of symbols considered in \cite{GRS}) since they are more suitable for the less symmetric perturbation which we consider.

\subsection{The main results}

We start by recalling the definition of the  {\it  Hausdorff distance} between any two compact sets $M,N\subset \R$: 
\[d_h(M,N)=\max\Big \{\sup_{\lambda\in M}{{\rm dist}}(\lambda,N),\; \sup_{\mu\in N}{\text{dist}}(\mu,M)\Big \}.\]

Our first main result gives a sharp upper bound on how much the spectra can "move" as sets:
\begin{theorem}\label{thm1}
 Consider the notation introduced in Hypothesis \ref{hyp1.2}. Then there exists $C>0$ such that \[d_h\big (\sigma(K_\delta),\, \sigma(K_{0})\big )\leq C\sqrt{|\delta|}\] for all $|\delta|\leq 1$. This bound is sharp, in the sense that one can construct a $K_0$ such that $0\in\sigma(K_0)$ while the spectrum of $K_\delta$ develops gaps of order $\sqrt{|\delta|}$ near zero. 
\end{theorem}

The next straightforward corollary spells out in a detailed way how the interior non-trivial gaps in the spectrum of $K_0$ may vary with $\delta$. 

\begin{corollary}\label{coro1}
Assume that $K_0$ has an open spectral gap $(\lambda_0,\mu_0)$ with $\lambda_0,\mu_0\in \sigma(K_0)$.  Then there exists a constant $C>0$ (the same as in Theorem \ref{thm1}), independent on the spectral gap, such that for all $|\delta|< (\mu_0-\lambda_0)^2/(4C^2)$ the interval $[\lambda_0+C\sqrt{|\delta|}, \mu_0-C\sqrt{|\delta|}]$ is non-empty and belongs to the resolvent set of $K_\delta$. Moreover, both intervals $[\lambda_0-C\sqrt{|\delta|},\lambda_0+C\sqrt{|\delta|}]$ and $[\mu_0-C\sqrt{|\delta|},\mu_0+C\sqrt{|\delta|}]$ have a non-empty intersection with $\sigma(K_\delta)$ for all $|\delta|< (\mu_0-\lambda_0)^2/(4C^2)$. 
\end{corollary}

The next main result states that the spectral edges of $K_\delta$ have a Lipschitz variation at  $\delta=0$: 

\begin{theorem}\label{thm2}
Let $E_+(\delta):=\sup\sigma(K_\delta)$ and $E_-(\delta):=\inf\sigma(K_\delta)$. There exists $C>0$ such that $|E_\pm(\delta)-E_\pm(0)|\leq C\,|\delta|$, for all $|\delta |\leq 1$. 
\end{theorem}

The next corollary describes the variation of the edges of those interior gaps which remain open at $\delta=0$, and gives a precise control with respect to the width of the spectral gap:
 
\begin{corollary}\label{coro2}
Consider the same setting and the same notation as in Corollary \ref{coro1}. Let $|\delta| <(\mu_0-\lambda_0)^2/(4C^2)$. Since $(\mu_0+\lambda_0)/2$ is in the resolvent set of $K_\delta$, and both sets $\sigma(K_\delta)\cap (-\infty,(\mu_0+\lambda_0)/2)$ and $\sigma(K_\delta)\cap ((\mu_0+\lambda_0)/2,\infty)$ are non-empty, we may define 
\[\lambda_\delta :=\sup \Big (\sigma(K_\delta)\cap (-\infty,(\mu_0+\lambda_0)/2)\Big ),\quad \mu_\delta:=\inf\Big (  \sigma(K_\delta)\cap ((\mu_0+\lambda_0)/2,\infty)\Big ).\]
 Then there exists a constant $\tilde{C}>0$, independent of $\mu_0-\lambda_0$, and some   $0<\delta_1<(\mu_0-\lambda_0)^2/(4C^2)$ such that 
\[\max\big \{|\lambda_\delta-\lambda_0|\, ,\,|\mu_\delta-\mu_0| \big \}\leq  \frac{\tilde{C}\,|\delta|}{\mu_0-\lambda_0},\quad \forall |\delta| \leq \delta_1.\]
\end{corollary}

\begin{remark}
Corollary \ref{coro2} is stronger than Corollary \ref{coro1} only when $\sqrt{|\delta|}$ is much smaller than the width of the gap $\mu_0-\lambda_0$. An important point is that the constant $C$ in Corollary \ref{coro1} is independent of the gap, while the Lipschitz constant in Corollary \ref{coro2} is inverse proportional with the width of the gap at $\delta=0$. This is compatible with Theorem \ref{thm1}: when $|\delta|$ increases and becomes of order $(\mu_0-\lambda_0)^2$, the gap might even close. 
\end{remark}

\begin{remark}
    When $F(x)=x$, the results of Theorem \ref{thm2} and Corollary \ref{coro2} are also obtained in \cite{GRS}. On the other hand, the results of Theorem \ref{thm1} and Corollary \ref{coro1} are new. We note that if one is only interested in proving Lipschitz behavior of the inner gap edges $\lambda_\delta$ and $\mu_\delta$, one does not need the explicit estimate in our Theorem \ref{thm1}, but only some a-priori knowledge of their continuity, as in   \cite{GRS}. 
\end{remark}

\section{Technical preliminaries}\label{sec2}
   
We consider that Lemma \ref{P-U-trsf} is a rather well-known fact and omit its proof.

\subsection{Known facts about the Hausdorff distance between spectra}

The following lemma is well-known but also very important, hence we prove it for completeness, see also \cite{CP2}.  
\begin{lemma}\label{lemma7}
Let $A$ and $B$ be self-adjoint and bounded. Let $E_+(A)=\sup\, \sigma(A)$, $E_-(A)=\inf\, \sigma(A)$, and $E_\pm(B)$ denotes the same for $B$. Then 
\[|E_\pm(A)-E_\pm(B)|\leq d_h\big (\sigma(A),\, \sigma(B)\big )\leq \Vert A-B\Vert.\]
\end{lemma}
\begin{proof}
Let us prove the first inequality but only for "$E_+$". Let us assume, without loss of generality, that $E_+(A)\leq E_+(B)$. Then \[0\leq E_+(B)-E_+(A)=\text{dist}\big (E_+(B),\,\sigma (A)\big)\leq \sup_{\lambda\in \sigma(B)}\text{dist}\big (\lambda,\,\sigma (A)\big)\leq d_h\big (\sigma(A),\, \sigma(B)\big ).\]

Now let us prove the second inequality. Let $z\not\in \sigma(A)$. We have 
\[B-z\,\bb1=\Big (\bb1 +(B-A)\, (A-z\,\bb1)^{-1}\Big )\, (A-z\,\bb1).\]
If $\text{dist}\big (z,\sigma(A)\big )>\Vert A-B\Vert$, then 
\[\Vert (B-A)\, (A-z\,\bb1)^{-1}\Vert <1 \]
and $z\not\in \sigma(B)$. This means that the spectrum of $B$ is located within a neighborhood of width $\Vert A-B\Vert$ of the spectrum of $A$.  The same conclusion holds for $A$ replaced with $B$.

\end{proof}

Another useful inequality is the following. 
\begin{lemma}
Let $A,\, B,\,C,\, D$ be bounded self-adjoint operators. Then
\begin{align}\label{hc31}
|E_\pm(A)-E_\pm(D)|\leq \Vert A-B\Vert +|E_\pm(B)-E_\pm(C)|+\Vert C-D\Vert.
\end{align}
\end{lemma}
\begin{proof}
Direct application of the triangle inequality and of Lemma \ref{lemma7}. 
\end{proof}

\subsection{Reduction to compact support in the second variable of the distribution kernel.} 
We refer to Hypothesis \ref{hyp1.2} for the notation involving $K_\delta$ and $\mathfrak{K}_\delta$. 
\begin{lemma}\label{lemmahc1}
Let $0\leq f\leq 1$ be smooth and compactly supported, with $f(x)=1$ in a neighborhood of $0$. Let $\tilde{K}_\delta$ be the operator with the integral kernel $\tilde{\mathfrak{K}}_\delta(u,v):=f(\sqrt{|\delta|}\, v)\, \mathfrak{K}_\delta(u,v)$. Then the symbol $\tilde{a}[F]_\delta$ of $\tilde{K}_\delta$ obeys \eqref{dc1} uniformly in $|\delta|\leq 1$ and 
\[\Vert K_\delta-\tilde{K}_\delta\Vert =\mathcal{O}(\delta^\infty).\]
\end{lemma}
\begin{proof}
We may assume $\delta\geq 0$. Denote by $\hat{f}$ the Fourier transform of $f$. The symbol of $\tilde{K}_\delta$ is a convolution:
\[\tilde{a}[F]_\delta(x,\xi)=(2\pi)^{-d/2}\int_{\R^d}a\big (x+F(\delta\, x)\,,\,\xi-\xi'\big )\frac{\hat{f}(\xi'/\sqrt{\delta})}{\delta^{d/2}}\; d\xi',\]
while the symbol of $K_\delta-\tilde{K}_\delta$ is 
\[(2\pi)^{-d/2}\int_{\R^d}\Big (a\big (x+F(\delta\, x)\,,\,\xi\big )-a\big (x+F(\delta\, x)\,,\,\xi-\xi'\big )\Big ) \frac{\hat{f}(\xi'/\sqrt{\delta})}{\delta^{d/2}}\; d\xi'\]
where we used that $f(0)=1$. Using the integral Taylor formula we have 
\begin{align*}
&a\big (x+F(\delta\, x)\,,\,\xi\big )-a\big (x+F(\delta\, x)\,,\,\xi-\xi'\big )=\int_0^1 dr\, \frac{d}{dr}a\big (x+F(\delta\, x)\,,\,\xi-\xi'+r\xi'\big )\\
&\quad =(\xi'\cdot \nabla_\xi) a\big (x+F(\delta\, x), \xi\big )-\int_0^1 dr\, r\, (\xi'\cdot \nabla_\xi)^2\, a\big (x+F(\delta\, x)\,,\,\xi-\xi'+r\xi'\big )\\
&\quad =\sum_{j=0}^{N-1}\frac{(-1)^{j}}{j!}(\xi'\cdot \nabla_\xi)^{j+1} a\big (x+F(\delta\, x), \xi\big )\\
&\quad \quad +\frac{(-r)^{N}}{N!}\int_0^1 dr\, (\xi'\cdot \nabla_\xi)^{N+1}\, a\big (x+F(\delta\, x)\,,\,\xi-\xi'+r\xi'\big )
\end{align*}
for every $N\geq 1$. Using that all the partial derivatives of $f$ at zero equal zero, we may write the symbol of $K_\delta-\tilde{K}_\delta$ as
\[(2\pi)^{-d/2}\int_0^1 dr \; \frac{(-r)^{N}}{N!}\int_{\R^d}( \xi'\cdot \nabla_\xi)^{N+1} a\big (x+F(\delta\, x)\,,\,\xi-\xi'+r\xi'\big ) \frac{\hat{f}(\xi'/\sqrt{\delta})}{\delta^{d/2}}\; d\xi',\]
for all $N\geq 1$. 
After a change of variables, this symbol reads as
\begin{align*}&\delta^{(N+1)/2}(2\pi)^{-d/2}\times \\
&\quad \times \int_0^1 dr \; \frac{(-r)^{N}}{N!}\int_{\R^d}( \xi'\cdot \nabla_\xi)^{N+1} a\big (x+F(\delta\, x)\,,\,\xi-\sqrt{\delta}\xi'+r\sqrt{\delta} \xi'\big ) \,\hat{f}(\xi')\; d\xi',
\end{align*}
for all $N\geq 1$. 

This symbol obeys \eqref{dc1} where the supremum is bounded by $C_{N,\alpha,\beta} \; \delta^{(N+1)/2}$ for all $N$. An application of the Calder{\' o}n-Vaillancourt Theorem \cite{CV} finishes the proof. 
\end{proof}

\subsection{A localization result.}

Let $0\leq g\leq 1$ be smooth with compact support such that 
\begin{align}\label{hc5}
    \sum_{\gamma\in \Z^d} g^2(x-\gamma)=1,\quad \forall x\in \R^d. 
\end{align}
Let $0\leq \tilde{g}\leq 1$ be any other compactly supported function and define $g_{\delta,\gamma}(x):=g(\sqrt{|\delta|}\; x-\gamma)$ and $\tilde{g}_{\delta,\gamma}(x):=\tilde{g}(\sqrt{|\delta|}\; x-\gamma)$. 

\begin{lemma}\label{lemma2}
Let $\mathcal{T}:=\{T_\gamma\}_{\gamma\in \Z^d}$ be any family of bounded operators on $L^2(\R^d)$ and let  $|||\mathcal{T}|||:=\sup_{\gamma\in \Z^d}\Vert T_\gamma\Vert$. We define 
\[\Gamma_{\tilde{g}}(\mathcal{T}):=\sum_{\gamma\in \Z^d}\tilde{g}_{\delta,\gamma}\; T_\gamma \; {g}_{\delta,\gamma}.\]
Then there exists a constant $C$ independent of $|\delta|\leq 1$ such that 
\[\Vert \Gamma_{\tilde{g}}(\mathcal{T})\Vert \leq C\; |||\mathcal{T}|||.\]
\end{lemma}
\begin{proof}
Given $\gamma\in \Z^d$ we denote by $V_\gamma$ the set of all $\gamma'\in \Z^d$ with the property that the support of  $\tilde{g}_{\delta,\gamma'}$ has a non-empty overlap with the support of $\tilde{g}_{\delta,\gamma}$, including $\gamma'=\gamma$. Denote by $\nu\in\mathbb{N}\setminus\{0\}$ the cardinal of $V_\gamma$; it is clearly independent of $\gamma$ and $\delta$. For $\psi\in L^2(\R^d)$:
\begin{align*}
    &\Vert \Gamma_{\tilde{g}}(\mathcal{T})\psi\Vert^2\\
    &\quad = \sum_{\gamma\in\Z^d}\sum_{\gamma'\in V_\gamma}\langle\tilde{g}_{\delta,\gamma}\; T_\gamma\; {g}_{\delta,\gamma} \; \psi\,,\, \tilde{g}_{\delta,\gamma'}\; T_\gamma\; {g}_{\delta,\gamma'}\;  \psi \rangle \leq \frac{\nu +1}{2}\sum_{\gamma\in\Z^d}\Vert T_\gamma\; {g}_{\delta,\gamma} \; \psi\Vert ^2\\
    &\quad \leq |||\mathcal{T}|||^2\; \frac{\nu +1}{2}\; \sum_{\gamma\in\Z^d} \int_{\R^d} g^2(\sqrt{|\delta|}\; x-\gamma)|\psi(x)|^2\; dx=|||\mathcal{T}|||^2\; \frac{\nu +1}{2}\; \Vert \psi\Vert^2,
\end{align*}
where in the last equality we used \eqref{hc5}. 

\end{proof}

\section{Proof of Theorem \ref{thm1}}

For simplicity, let $0\leq \delta\leq 1$. From Lemma \ref{lemmahc1} and Lemma \ref{lemma7} we infer that the Hausdorff distance between the spectra of $K_\delta$ and $\tilde{K}_\delta$ is of order $\delta^\infty$. Let us define the operator $\mathring{K}_\delta$ through its integral kernel given by $\mathring{\mathfrak{K}}_\delta(u,v)=f(\sqrt{\delta}\; v)\mathfrak{K}_0(u,v)$. Then with the same proof as in Lemma \ref{lemmahc1} one can show that $\Vert \mathring{K}_\delta-K_0\Vert=\mathcal{O}(\delta^\infty)$, and the same is true for the Hausdorff distance between their spectra. Therefore, according to the second inequality in \eqref{hc31}, it is enough to prove that the Hausdorff distance between the spectra of $\tilde{K}_\delta$ and $\mathring{K}_\delta$ is of order $\sqrt{\delta}$.

 Let $\tau_\alpha$ be the unitary operator induced by the translation with $-\alpha$, i.e. $(\tau_\alpha \psi)(x)=\psi(x-\alpha)$; we use the notations introduced in Lemma \ref{lemma2} and work with $\Gamma_g$, i.e. with $\tilde{g}=g$. We shall prove the following statement.
 
 \begin{proposition}\label{prop1}
 Let $\z\in\mathbb{C}$ be in the resolvent set of $\mathring{K}_\delta$ defined above. Let us define: \[T_\gamma(\z)=\tau_{-F(\sqrt{\delta}\, \gamma)}(\mathring{K}_\delta-\z\bb1)^{-1}\tau_{F(\sqrt{\delta}\, \gamma)}\] the associated family $\mathcal{T}(\z):=\big\{T_\gamma(\z)\big\}_{\gamma\in\Z^d}$ as in Lemma \ref{lemma2} and the "remainder operator":
 \[R_\delta(\z):=(\tilde{K}_\delta-\z {\bb1})\Gamma_g(\mathcal{T}(\z))-\bb1.\]
 Then there exists a constant $C$ such that for all $0\leq \delta\leq 1$ we have 
 
 \begin{align}\label{hc30}
 \Vert R_\delta(\z)\Vert \leq C\;\frac{\sqrt{\delta}}{ \text{dist}(\z,\sigma(\mathring{K}_\delta))}\; .
 \end{align}
 In particular, this implies that the spectrum of $\tilde{K}_\delta$ belongs to a neighbourhood of width $C\, \sqrt{\delta}$ of the spectrum of $\mathring{K}_\delta$. 
\end{proposition}

\subsection{Proof of Proposition \ref{prop1}}

Here we use some of the ideas employed in \cite{BBC}. We need to investigate the operator $\tilde{K}_\delta\,  g_{\delta,\gamma}$ and compare it with $g_{\delta,\gamma}\, \tau_{-F(\sqrt{\delta}\gamma)}\,\mathring{K}_\delta\, \tau_{F(\sqrt{\delta}\gamma)}$. If they were equal to each other, then $R_\delta(\z)$ would equal zero. In fact there are two contributions to $R_\delta(\z)$: one coming from replacing $\tilde{K}_\delta\,  g_{\delta,\gamma}$ with $ \tau_{-F(\sqrt{\delta}\gamma)}\,\mathring{K}_\delta\, \tau_{F(\sqrt{\delta}\gamma)}\, g_{\delta,\gamma}$, and the other one coming from the commutator $ [\tau_{-F(\sqrt{\delta}\gamma)}\,\mathring{K}_\delta\, \tau_{F(\sqrt{\delta}\gamma)}\,,\, g_{\delta,\gamma}]$.

\begin{lemma}\label{lemma3}
There exists a constant $C>0$ such that 
\[\Vert \tilde{K}_\delta\,  g_{\delta,\gamma}-\tau_{-F(\sqrt{\delta}\gamma)}\,\mathring{K}_\delta\, \tau_{F(\sqrt{\delta}\gamma)}\,  g_{\delta,\gamma}\Vert \leq C\, \sqrt{\delta},\quad \forall \gamma\in \Z^d.\]
\end{lemma}
\begin{proof}
The distribution kernel of the operator \[L_{\delta,\gamma}:=\tilde{K}_\delta\,  g_{\delta,\gamma}-\tau_{-F(\sqrt{\delta}\gamma)}\,\mathring{K}_\delta\, \tau_{F(\sqrt{\delta}\gamma)}\,  g_{\delta,\gamma}\] from the statement of the Lemma ig given by:
\[f(\sqrt{\delta}\,v) \Big (\mathfrak{K}_0\big (u+F(\delta\; u)\,,\,v\big )-\mathfrak{K}_0\big (u+F(\sqrt{\delta}\,\gamma)\,,\,v \big )\Big ) g\big(\sqrt{\delta}\, (u-v/2)-\gamma\big).\]
Let us denote by 
\[h_{y}(x):=f(x)g(-x/2 +y).\]
We see that $h_y$ is identically zero if $|y|$ is large enough. We may find some smooth and compactly supported function $0\leq \tilde{h}\leq 1$ such that 
\[h_y(x)\tilde{h}(y)=h_y(x),\quad \forall x,y\in \R^d.\]
The role of $y$ is played by $\sqrt{\delta}\,u-\gamma$, which means that the quantity $|\sqrt{\delta}\, u-\gamma|$ remains uniformly bounded in both $\delta$ and $\gamma$ due to the presence of $\tilde{h}$. 

Denote by $\nabla_1 a(x,\xi)$ the partial gradient of $a(x,\xi)$ with respect to the spatial variables $x$. Then by denoting with 
\begin{align*}
 & \alpha_{\delta,\gamma}(u,\xi):= a\big (u+F(\delta\; u)\,,\,\xi\big )-a\big (u+F(\sqrt{\delta}\,\gamma)\,,\,\xi \big )\\
  &=\Big (F(\delta\, u)-F(\sqrt{\delta}\,\gamma)\Big )\cdot \int_0^1d\,r\,   \nabla_1a\big (u+F(\sqrt{\delta}\,\gamma)+\, r\,[F(\delta\; u)-F(\sqrt{\delta}\,\gamma)]\,,\,\xi\big ),
\end{align*}
the symbol of our operator $L_{\delta,\gamma}$ is 
\[b_{\delta,\gamma}(u,\xi):=(2\pi)^{-d/2}\tilde{h}(\sqrt{\delta}\, u-\gamma)\,\int_{\R^d}
\, \alpha_{\delta,\gamma}(u,\xi-\xi')\, \frac{\hat{h}_{\sqrt{\delta}\,t-\gamma}(\xi'/\sqrt{\delta})}{\delta^{d/2}}\, d\xi'.\]
An important observation is that the function 
\[\tilde{h}(\sqrt{\delta}\, u-\gamma)\,\Big (F(\delta\, u)-F(\sqrt{\delta}\,\gamma)\Big )\]
is uniformly bounded by $\sqrt{\delta}$ together with all its derivatives. 
Thus all the seminorms of the above symbol will have (at least) a factor $\sqrt{\delta}$,  uniformly in $\gamma$ and we are done after an application of Calder{\' o}n-Vaillancourt. 

\end{proof}

\begin{lemma}\label{lemma4}
There exists a constant $C>0$ such that 
\[\Vert   \big [g_{\delta,\gamma},\; \tau_{-F(\sqrt{\delta}\,\gamma)}\,\mathring{K}_\delta\, \tau_{F(\sqrt{\delta}\,\gamma)}  \big ]\Vert \leq C\, \sqrt{\delta},\quad \forall \gamma\in \Z^d.\]
\end{lemma}
\begin{proof}
The distribution kernel of the above commutator is 
\[f(\sqrt{\delta}\,v)\mathfrak{K}_0\big (u+F(\sqrt{\delta}\,\gamma)\,,\,v \big )\Big ( g(\sqrt{\delta}\, (u+v/2)-\gamma)-g(\sqrt{\delta}\, (u-v/2)-\gamma)\Big )\]
which equals 
\[\sqrt{\delta}\,2^{-1}\, \int_{-1/2}^{1/2} dr\, f(\sqrt{\delta}\,v)\mathfrak{K}_0\big (u+F(\sqrt{\delta}\,\gamma)\,,\,v \big ) \,v\cdot \nabla g(\sqrt{\delta}\, u-\gamma +r\sqrt{\delta}\,v/2).\]
The $v$ appearing in the factor $v\cdot \nabla g$ has to be coupled with $\mathfrak{K}_0$, in the sense that when we write the symbol of the commutator as a convolution, the factor $\mathfrak{K}_0\big (u+F(\sqrt{\delta}\,\gamma)\,,\,v \big ) \,v$ becomes  $\nabla_\xi a(u+ F(\sqrt{\delta}\,\gamma)\,,\,\xi-\xi')$ in the convolution. It turns out that again, all the seminorms of the commutator symbol will be of order $\sqrt{\delta}$ uniformly in $\gamma$. 

\end{proof}

We are now ready to complete the proof of Proposition \ref{prop1}. Let $\Omega\subset \R^d$ be a ball which contains the set 
\[\{x\in \R^d\;|\; {\rm dist}\big (x,{\rm supp}(g)\big )\leq 1\}\] and let us denote by $\tilde{g}$ the indicator function of $\Omega$.
\begin{remark}
Due to our choice of $f$ with support in the unit ball, the presence of $f(\sqrt{\delta}\, v)$ in the distribution kernels of both $\tilde{K}_\delta$ and $\mathring{K}_\delta$ implies that: 
\begin{equation}\label{hc6}
\begin{split}
\tilde{K}_\delta\,  g_{\delta,\gamma}=\tilde{g}_{\delta,\gamma}\tilde{K}_\delta\,  g_{\delta,\gamma},\quad \tau_{-F(\sqrt{\delta}\gamma)}\,\mathring{K}_\delta\, \tau_{F(\sqrt{\delta}\gamma)}\, 
\\ g_{\delta,\gamma}=\tilde{g}_{\delta,\gamma}\tau_{-F(\sqrt{\delta}\gamma)}\,\mathring{K}_\delta\, \tau_{F(\sqrt{\delta}\gamma)}\,  g_{\delta,\gamma}.
\end{split}
\end{equation}
\end{remark}

Now let $\{M_\gamma(\z)\}_{\gamma\in \Z^d}$ be a family $\mathcal{M}(\z)$ of operators given by 
\begin{align*}&M_\gamma(\z)\\
&=\Big (\tilde{K}_\delta\,  -\tau_{-F(\sqrt{\delta}\gamma)}\,\mathring{K}_\delta\, \tau_{F(\sqrt{\delta}\gamma)}\Big )g_{\delta,\gamma}\, T_\gamma(\z)-\big [g_{\delta,\gamma},\; \tau_{-F(\sqrt{\delta}\gamma)}\,\mathring{K}_\delta\, \tau_{F(\sqrt{\delta}\gamma)}  \big ] \, T_\gamma(\z).
\end{align*}
A short computation using also \eqref{hc6} shows that 
\[R_\delta(\z)=\Gamma_{\tilde{g}}\big ( \mathcal{M}(\z)\big ),\]
hence an application of Lemma \ref{lemma2} with $\mathcal{T}$ replaced by $\mathcal{M}$ 
finishes the proof of \eqref{hc30}. 

Now let us investigate the spectral consequences. If $\z\in\mathbb{C}$ is in the resolvent set of $\mathring{K}_\delta$, we have the identity
\[(\tilde{K}_\delta -\z\, \bb1)\Gamma_g\big (\mathcal{T}(\z)\big )=\bb1+R_\delta(\z).\]
Now if we also impose that ${\rm dist}\big (\z,\sigma(\mathring{K}_\delta)\big )> C\, \sqrt{\delta}$, then $\Vert R_\delta(z)\Vert <1$ and $(\tilde{K}_\delta -\z\, \bb1)$ is invertible, thus $\z$ cannot belong to the spectrum of $\sigma(\tilde{K}_\delta)$. Thus if $\lambda\in \sigma(\tilde{K}_\delta)$, then ${\rm dist}\big (\z,\sigma(\mathring{K}_\delta)\big )\leq  C\, \sqrt{\delta}$. This ends the proof of 
Proposition \ref{prop1}. 

\qed

\subsection{Concluding the proof of Theorem \ref{thm1}}\label{sec3.2}

We have seen in Proposition \ref{prop1} that if $\z\in\mathbb{C}$ is at a distance larger than a constant times $\sqrt{\delta}$ from the spectrum of $\mathring{K}_\delta$, then $\z$ is also in the resolvent set of $\tilde{K}_\delta$. Exactly the same type of proof can be used when we swap the roles of $a\big (x+F(\delta\, x),\xi\big )$ and  $a(x,\xi)$, namely by putting $\tilde{a}(x,\xi):=a\big (x+F(\delta\, x),\xi\big )$ and $\tilde{a}_\delta(x,\xi):=\tilde{a}\big (x-F(\delta \,x),\xi\big )$. Thus this proves that the Hausdorff distance between $\sigma(\mathring{K}_\delta)$ and $\sigma(\tilde{K}_\delta)$ goes like $\sqrt{\delta}$.  

This bound cannot be made better in general. Let $d=2$ and let \[a_\delta(x,\xi)=\cos(\xi_1)+\cos\big (\xi_2+(1+\delta)x_1\big ).\]
Through Weyl quantization, this symbol generates an operator which is isospectral with the Hofstadter model, in the Landau gauge, with a constant magnetic field $b=1+\delta$. It is known \cite{HS1989, HeSj} that $\Op(a_0)$ corresponds to the "half-flux case", and the operator has an absolutely continuous gap-less spectrum which contains the origin. If $\delta\neq 0$ is small, then the spectrum of $\Op(a_\delta)$ develops gaps near zero of width $\sqrt{\delta}$. A recent detailed analysis regarding the magnetic perturbations of "Dirac cones" which produce gaps of order $\sqrt{\delta}$ may be found in \cite{CHP}.

\section{Proof of Theorem \ref{thm2}}
We only prove the theorem for $E_+$ and $0\leq \delta\leq 1$. Due to \eqref{hc31}, it is enough to prove the statement for the pair of operators $\tilde{K}_\delta$ and $\mathring{K}_\delta$, where as before, $\tilde{K}_\delta$ corresponds to the integral kernel $f(\sqrt{\delta}\, v)\mathfrak{K}_0(u+F(\delta\,u\,),\,v)$, while $\mathring{K}_\delta$  corresponds to the integral kernel $f(\sqrt{\delta}\, v)\mathfrak{K}_0(u,v)$. 

\begin{lemma}\label{lemma5}
Let $\delta>0$. For every  $\psi\in L^2(\R^d)$ we define 
\[\Psi_\delta(x,y):=\left  (\frac{\delta}{4\pi}\right )^{d/4}\, e^{-\delta\frac{|x-y|^2}{8}}\; \psi(x).\]
Then $\Psi_\delta\in L^2(\R^{2d})$ and $\Vert \Psi_\delta\Vert_{L^2(\R^{2d})}=\Vert \psi\Vert_{L^2(\R^{d})}$. 
\end{lemma}
\begin{proof}
Direct computation. 

\end{proof}

\begin{lemma}\label{lemma6}
Let $x,y,z\in\R^d$. Then: 
\[2^{-1}\, |x+y/2 -z|^2 + 2^{-1}\, |x-y/2 -z|^2-y^2/4=|x-z|^2.\]
\end{lemma}
\begin{proof}
Direct computation (the parallelogram identity). 

\end{proof}

For $\psi\in\mathscr{S}(\R^d)$ we notice that we have the following identity:
\begin{equation}\label{hc7}
\begin{split}
    \langle\psi,\, \tilde{K}_\delta \psi\rangle&=\int_{\R^{2d}}du\, dv\, \overline{\psi}(u+v/2)\, f(\sqrt{\delta}\, v)\mathfrak{K}_0\big (u+F(\delta\,u)\,,\,v\big )\, {\psi}(u-v/2) \\ 
    &=\int_{\R^d}dy\int_{\R^{2d}}du\, dv\, \overline{\psi}(u+v/2)\, f(\sqrt{\delta}\, v)\mathfrak{K}_0\big (u+F(\delta\,u)\,,\,v\big )\,\times \\
    &\qquad \qquad \times \left  (\frac{\delta}{4\pi}\right )^{d/2}\, e^{-\delta\frac{|u-y|^2}{4}}\,  {\psi}(u-v/2),
    \end{split}
\end{equation}
where we used that the $y$-integral of the heat kernel equals $1$.

The next lemma is very important, and says that we may replace $F(\delta\, u)$ with $F(\delta\, y)$, making only an error of order $\delta$:
\begin{lemma}\label{lemma8}
There exists $C>0$ such that 
\begin{align*}
    \langle\psi,\, \tilde{K}_\delta \psi\rangle & \leq \int_{\R^d}dy\int_{\R^{2d}}du\, dv\, \overline{\psi}(u+v/2)\, f(\sqrt{\delta}\, v)\mathfrak{K}_0\big (u+F(\delta\,y)\,,\,v\big )\,\times  \\
    &\qquad \qquad \times \left  (\frac{\delta}{4\pi}\right )^{d/2}\, e^{-\delta\frac{|u-y|^2}{4}}\,  {\psi}(u-v/2)+ C\, \delta\, \Vert\psi\Vert^2,\quad 0< \delta\leq 1.
\end{align*}
\end{lemma}
\begin{proof}
Let us consider the following distribution kernel:
\[\int_{\R^d}dy\, f(\sqrt{\delta}\, v)\Big (\mathfrak{K}_0\big (u+F(\delta\,y)\,,\,v\big )-\mathfrak{K}_0\big (u+F(\delta\,u)\,,\,v\big )\Big )\,\left  (\frac{\delta}{4\pi}\right )^{d/2}\, e^{-\delta\frac{|u-y|^2}{4}}.\]
Denoting by $\nabla_1$ the partial gradient with respect to the "$u\in\R^d$" variables, we can write the above distribution kernel as:
\begin{align*}
  &\int_{\R^d}dy\,f(\sqrt{\delta}\, v)\, \big (F(\delta\, y)-F(\delta\, u)\big )\cdot  \nabla_1\,\mathfrak{K}_0\big (u+F(\delta\,u)\,,\,v\big )\,\left  (\frac{\delta}{4\pi}\right )^{d/2}\, e^{-\delta\frac{|u-y|^2}{4}}  \\
  & + \int_{\R^d}dy\,\int_0^1dr\, (1-r) \Big (\big (F(\delta\, y)-F(\delta\, u)\big ) \cdot  \nabla_1\Big )^2\, \times 
  \\
  &\qquad \qquad\times f(\sqrt{\delta}\, v)\,\mathfrak{K}_0\big (u+F(\delta\,u)+r\,( F(\delta\,y)-F(\delta\, u)),\,v\big ) \left  (\frac{\delta}{4\pi}\right )^{d/2}\, e^{-\delta\frac{|u-y|^2}{4}}.
\end{align*}
From our Hypothesis \ref{hyp1.2} we have $|F(\delta\,y)-F(\delta\, u)|\leq\,C\delta\, |u-y|$, hence both above kernels correspond to $S_{0,0}^0$ symbols due to the fact that the growth in $|u-y|$ is controlled by the Gaussian factor.  

Moreover, in the second kernel we can couple one power of $\delta$ with the quadratic term $|y-u|^2$ and thus we can bound this second kernel by a constant times $\delta$ and conclude that all the seminorms of its associated symbol will be of order $\delta$. 

For the first kernel we use the Taylor expansion: 
\begin{align}\label{hc40}
F(\delta\,y)-F(\delta\, u)=\delta\, \nabla F(\delta\,u)\cdot (y-u) +\mathcal{O}(\delta^2\,|y-u|^2).
\end{align}
The remarkable fact is that the linear term vanishes identically after integration in $y$. The quadratic term can be dealt with as we did with the second kernel concluding that it will generate an operator with norm of order $\delta$.  
\end{proof}

Using the notation from Lemmas \ref{lemma5} and \ref{lemma6}, the inequality from Lemma \ref{lemma8} reads as:
\begin{align}\label{hc8}
    \langle\psi,\, \tilde{K}_\delta \psi\rangle & \leq \int_{\R^d}dy\int_{\R^{2d}}du\, dv\, \overline{\Psi_\delta}(u+v/2,y)\, f(\sqrt{\delta}\, v)\, e^{\frac{\delta\, v^2}{16}}\,\mathfrak{K}_0(u+F(\delta\,y)\,,\,v)\, \times \nonumber \\
    &\quad \times {\Psi_\delta}(u-v/2,y)+ C\, \delta\, \Vert\psi\Vert^2.
\end{align}
Another crucial observation is that the operator with the integral kernel given by 
\[f(\sqrt{\delta}\, v)\, e^{\frac{\delta\, v^2}{16}}\,\mathfrak{K}_0(u+F(\delta\, y)\,,\,v),\quad \forall y\in\R^d,\]
appearing in \eqref{hc8}, is unitarily equivalent, by a conjugation with $\tau_{F(\delta\,y)}$, with the operator denoted from now on by $M_\delta$ given by the distribution kernel 
\[\mathfrak{M}_\delta(u,v)=f(\sqrt{\delta}\, v)\, e^{\frac{\delta\, v^2}{16}}\,\mathfrak{K}_0(u,v).\]
These operators have the same spectrum, for all $y\in\R^d$, 
thus from \eqref{hc8} and Lemma \ref{lemma5} we have 
\begin{align}\label{hc9}
    \langle\psi,\, \tilde{K}_\delta \psi\rangle & \leq E_+(M_\delta)\,\Vert\psi\Vert^2 + C\, \delta\, \Vert\psi\Vert^2.
\end{align}
Finally, we see that the operator $M_\delta-\mathring{K}_\delta$ has the integral kernel  
\[(\delta/16) \int_0^1dr \, f(\sqrt{\delta}\, v)\, e^{\frac{r\delta\, v^2}{16}}\,v^2\,\mathfrak{K}_0(u,v).\]
The factor $v^2$ multiplied by $\mathfrak{K}_0$ will generate (by the usual integration by parts procedure for oscillatory integrals) some second order derivatives in $\xi$ of the symbol $a(x,\xi)$, while the Gaussian is just a smooth function depending on $\sqrt{\delta}\, v$, which on the support of $f$ remains bounded. Hence the operator $M_\delta-\mathring{K}_\delta$ has a norm bounded by $\delta$, which together with \eqref{hc9} implies: 
\[\langle\psi,\, \tilde{K}_\delta \psi\rangle  \leq E_+(\mathring{K}_\delta)\,\Vert\psi\Vert^2 + C\, \delta\, \Vert\psi\Vert^2,\]
i.e. $E_+(\tilde{K}_\delta)\leq E_+(\mathring{K}_\delta)+C\, \delta\,$. The inequality where $\tilde{K}_\delta$ and $\mathring{K}_\delta$ exchange places can be proved in a similar way. 

\qed

\section{Proof of Corollaries \ref{coro1} and \ref{coro2}}

Corollary \ref{coro1} is just a direct consequence of the definition of the Hausdorff distance. For Corollary \ref{coro2} we use a similar trick with the one used in \cite{GRS}. Let $\gamma_0=(2\lambda_0+\mu_0)/3$ and let us define 
\begin{equation}\label{F-R1}
T_\delta:= (K_\delta-\gamma_0\,\bb1)^2=K_\delta^2-2\gamma_0 K_\delta+\gamma_0^2\, \bb1. 
\end{equation}

Let us assume for the moment that $E_\pm(T_\delta)$ are Lipschitz at $\delta=0$, a fact which we will prove later. 
If $\delta\geq 0$ is small enough, then by using the spectral theorem, the fact that $\gamma_0$ is closer to $\lambda_0$ than to $\mu_0$, and the a-priori estimate from Theorem \ref{thm1} which says that $|\lambda_\delta-\lambda_0|\leq C\,\sqrt{\delta}$, we have that \[E_-(T_\delta)=(\lambda_\delta-\gamma_0)^2.\]
The Lipschitzianity of $E_-(T_\delta)$ implies the  existence of a constant $C_1>0$ such that:
\[\Big | (\lambda_\delta-\gamma_0)^2-(\lambda_0-\gamma_0)^2\Big |\leq C_1\, \delta,\quad \forall\, 0\leq \delta<\delta_0.\]
Writing \[(\lambda_\delta-\gamma_0)^2-(\lambda_0-\gamma_0)^2=(\lambda_\delta-\lambda_0)\, (\lambda_\delta+\lambda_0-2\gamma_0)\]
we have some small enough $\delta_1<\delta_0$ such that 
\[|\lambda_\delta-\lambda_0|\leq \frac{C_1\,\delta}{|\lambda_\delta+\lambda_0-2\gamma_0|}\leq  \frac{C_1\,\delta}{\gamma_0-\lambda_0}=\frac{3C_1\,\delta}{\mu_0-\lambda_0},\quad \forall\, 0\leq \delta<{\delta}_1\]
and we are done. 

The only thing which remains to be proved is that $E_\pm(T_\delta)$ are Lipschitz at $\delta=0$. We start with a lemma.

\begin{lemma}\label{lemma9}Given  $B=\Op^w(b)$ with $ b\in S_{0,0}^0(\R^{2d})$, we  denote by $\mathcal{M}_\delta(B)$ the Weyl quantization of the perturbed symbol $b(x+\, F(\delta\, x)\,,\,\xi) $. 
If $a\in S_{0,0}^0(\R^{2d})$ then $K_\delta=\mathcal{M}_\delta(K_0)$ and 
\[\Vert \mathcal{M}_\delta (K_0^2)-K_\delta^2\Vert \leq C\, \delta.\]
\end{lemma}
\begin{proof} 
Denote by $\chi$ the indicator function of the unit hypercube $\Omega:=[-1/2,1/2]^d$. The operator $K_\delta$ can be seen as an operator in $\bigoplus_{\gamma\in \Z^d}L^2(\Omega)$ given by the operator valued matrix 
\begin{align*}
&A_{\gamma,\gamma'}(\delta)(\x,\x')\\
&\qquad :=\mathfrak{H}_0\big ((\x+\x'+\gamma+\gamma')/2 +F(\delta\, (\x+\x')/2+\delta (\gamma+\gamma')/2)\,,\,\x+\gamma-\x'-\gamma'\big ),
\end{align*}
where $(\x,\x')\in\Omega\times\Omega$ and $(\gamma,\gamma')\in\Z^d\times\Z^d$. 

Due to the strong localization of $\mathfrak{H}_0$ with respect to $v=x-x'$, one can prove that for every $N\geq 1$ there exists $C_N>0$ such that 
\[\Vert A_{\gamma,\gamma'}(\delta)\Vert_{L^2(\Omega)} \leq C_N \, \langle \gamma-\gamma'\rangle^{-N}.\]

For $(\x,\x')\in\Omega\times\Omega$ let us define:
\[\tilde{A}_{\gamma,\gamma'}(\delta)(\x,\x'):=\mathfrak{H}_0\big ((\x+\x'+\gamma+\gamma')/2 +F(\delta (\gamma+\gamma')/2)\,,\,\x+\gamma-\x'-\gamma'\big ).\]

Using a Taylor expansion for $F$ together with the strong localization of $\mathfrak{H}_0$ in the $v$ variable, one may also show that for
every $N\geq 1$ there exists $C_N>0$ such that 
\[\Vert A_{\gamma,\gamma'}(\delta)-\tilde{A}_{\gamma,\gamma'}(\delta)\Vert_{L^2(\Omega)} \leq C_N \, \delta \,\langle \gamma-\gamma'\rangle^{-N}.\]
Up to a use of the Schur test in $\bigoplus_{\gamma\in \Z^d}L^2(\Omega)$, we get that \[K_\delta -\sum_{\gamma,\gamma'\in \Z^d} \chi(\cdot +\gamma)\tilde{A}_{\gamma,\gamma'}(\delta)\, \chi(\cdot+\gamma')=\mathcal{O}(\delta).\] 
Thus up to an error of order $\delta$ in operator norm, we have that $K_\delta^2$ is given by 
\[\sum_{\gamma,\gamma'\in \Z^d} \sum_{\gamma''\in \Z^d}\chi(\cdot +\gamma)\tilde{A}_{\gamma,\gamma''}(\delta)\, \tilde{A}_{\gamma'',\gamma'}(\delta)\chi(\cdot+\gamma').\]
By replacing $F(\delta(\gamma+\gamma'')/2)$ and $F(\delta(\gamma''+\gamma')/2)$ with $F(\delta(\gamma+\gamma')/2)$, we produce an error of the type 
\[\delta \, \langle \gamma'-\gamma''\rangle\, \langle \gamma-\gamma''\rangle,\]
that is controlled by the strong off-diagonal decay in both $|\gamma-\gamma''|$ and $|\gamma''-\gamma'|$ induced by $\mathfrak{K}_0$. Hence $K_\delta^2$ is up to an error of order $\delta$ given by the operator valued matrix:
\begin{align*}B_{\gamma,\gamma'}(\delta)(\x,\x')&:=\big ({\rm Integral\, kernel\, of }K_0^2\big )\big ((\x+\x'+\gamma+\gamma')/2 \\
&\qquad +F(\delta (\gamma+\gamma')/2)\,,\,\x+\gamma-\x'-\gamma'\big ).
\end{align*}
Finally, by again using a Taylor expansion and a Schur test, one shows that this operator and $\mathcal{M}_\delta(K_0^2)$ differ from each other   by something of order $\delta$ in the operator topology, and the proof is finished.

\end{proof}

Going back to \eqref{F-R1}, we notice that Lemma \ref{lemma9} implies that modulo an error of order $\delta$ we can replace $T_b$ by 
\[\mathcal{M}_\delta(K_0^2)-2\gamma_0\,\mathcal{M}_\delta(K_0) +\gamma_0^2\,\bb1=\mathcal{M}_\delta(K_0^2-2\gamma_0\, K_0 +\gamma_0^2\,\bb1),\]
which is the same type as $K_\delta$ and thus the Lipschitzianity of $E_-(T_b)$ will follow from Theorem \ref{thm2}.

\end{document}